\documentclass[runningheads]{llncs}
\usepackage[T1]{fontenc}
\usepackage[utf8]{inputenc}
\usepackage{amsmath}
\usepackage{amssymb}
\usepackage{graphicx}
\usepackage{xcolor}
\usepackage{url}
\usepackage{subfigure}
\usepackage{multicol}
\usepackage{latexsym,amsmath,amssymb,stmaryrd,mathtools}
\usepackage{graphicx,epsfig,tikz}
\usepackage[normalem]{ulem}
\usepackage{float}
\usetikzlibrary{shapes.geometric,positioning,chains,fit,calc}
\tikzset{near start abs/.style={xshift=1cm},
	node/.style={circle,draw},
	nodeone/.style={circle,draw,gray, line width=0.7mm},
	nodered/.style={circle,draw,red, line width=0.7mm}}

\def\addlegendimage{\csname pgfplots@addlegendimage\endcsname}

\usepackage[useregional]{datetime2}

\newcommand{\LCP}{\textit{LCP}}
\newcommand{\LCS}{\textit{LCS}}

\usepackage{algorithm}
\PassOptionsToPackage{noend}{algpseudocode}% comment out if want ends to show
\usepackage{algpseudocode}

\begin{document}
\title{Longest Common Prefix Arrays for Succinct $k$-Spectra\thanks{Supported in part by the Academy of Finland via grants 339070 and 351150.}}
\author{Jarno N. Alanko \and
Elena Biagi \and
Simon J. Puglisi\orcidID{0000-0001-7668-7636}}
\authorrunning{J. N. Alanko, E. Biagi, S. J. Puglisi}
\institute{
Helsinki Institute for Information Technology\\
Department of Computer Science, University of Helsinki, Finland\\
\email{jarno.alanko@helsinki.fi}}
\maketitle              % typeset the header of the contribution
\begin{abstract}
The $k$-spectrum of a string is the set of all distinct substrings of length $k$ occurring in the string. $K$-spectra have many applications in bioinformatics including pseudoalignment and genome assembly.
The \emph{Spectral Burrows-Wheeler Transform} (SBWT) has been recently introduced as an algorithmic tool to efficiently represent and query these objects.
The longest common prefix ($\LCP$) array for a $k$-spectrum is an array of length $n$ % why m and not n?
that stores the length of the longest common prefix of adjacent $k$-mers as they occur in lexicographical order. The $\LCP$ array has at least two important applications, namely to accelerate pseudoalignmet algorithms using the SBWT and to allow simulation of variable-order de Bruijn graphs within the SBWT framework. In this paper we explore algorithms to compute the $\LCP$ array efficiently from the SBWT representation of the $k$-spectrum. Starting with a straightforward $O(nk)$ time algorithm, we describe algorithms that are efficient in both theory and practice. We show that the $\LCP$ array can be computed in optimal $O(n)$ time, where $n$ is the
%number of $k$-mers in the $k$-spectrum Not exactly right. $n$ is the number of $k$-mers in the \emph{extended spectrum}, which contains the dummies, and can be up to $k$ times larger than the original spectrum. Suggested alternative phrasing: $n$ is the 
length of the SBWT of the spectrum. In practical genomics scenarios, we show that this theoretically optimal algorithm is indeed practical, but is often outperformed on smaller values of $k$ by an asymptotically suboptimal algorithm that interacts better with the CPU cache. Our algorithms share some features with both classical Burrows-Wheeler inversion algorithms and LCP array construction algorithms for suffix arrays.

\end{abstract}
%
%
%

%%%%%%%%%%%%%%%%%%%%%%%%%%%%%%%%%%%%%%%%%%%%%%%%%%%%%%%%%%%%%%%%%%%%%%%%%%%%%%%%%
\section{Introduction}

The {\em $k$-spectrum} of a string $S$ is the set of substrings of a given length $k$ that occur in $S$.
Indexing $k$-spectra has become an important topic in bioinformatics, perhaps most notably in the form of de Bruijn graphs, which are a long-standing tool for genome assembly~\cite{compeau2011bruijn} and more recently for pangenomics~\cite{AVMP23,holley2020bifrost,marchet2021data}. In metagenomics, $k$-spectra are used as concise approximation of the sequence content of the sample, allowing rapid similarity estimation between data collected from sequencing runs~\cite{maillet2012compareads,ondov2016mash}. In most current genomics applications $k$ is in the range from 20 to 100.

Recently, the Spectral Burrows-Wheeler transform (SBWT)~\cite{ACDA2023} has been introduced as an efficient way to losslessly encode and query $k$-spectra. In particular, the SBWT encodes the $k$-mers of the spectrum in {\em colexicographical} order. Combining the SBWT with entropy compressed bitvectors leads to a data structure that encodes the spectrum in little more than 2 bits per $k$-mer~\cite{ACDA2023,ABPV23}. Remarkably, while in this form, it is also possible to answer lookup queries on the spectrum rapidly, in fact in $O(k)$ time.

The SBWT allows a lookup query for a given $k$-mer to be reduced to at most $k$ {\em right-extension queries}. The input to a right-extension query is a letter $c$ and an interval $[i,j]$ in the colexicographic ordering of the $k$-mers of the spectrum such that all $k$-mers in the interval share a suffix $X$. The query returns the interval $[i',j']$ that contains all the $k$-mers that have $Xc$ as a suffix (or an empty interval if none do). 

\medskip

Our focus in this paper is on augmenting the SBWT with a data structure called the longest common suffix (LCS) array that stores the lengths of the longest common suffixes of adjacent $k$-mers in colexicographical order (we give a precise definition below\footnote{We remark here that the LCS array of a colexicographically-ordered spectrum is equivalent to the longest common prefix (LCP) array of the lexicographically-ordered spectrum, and the algorithms we describe in this paper to compute the LCS array are trivially adapted to compute the LCP array.}). The $\LCS$ array allows us to support so-called  {\em left contraction} queries: given an interval $[i,j]$ in the colexicographical ordering of the $k$-mers of the spectrum containing all the $k$-mers that share a suffix $X$ of length $k' \in (0, k]$
and a contraction point $t<k'$, a left contraction returns the interval $[i',j']$ containing all the $k$-mers having $X[t..k']$ as a suffix. Left contractions have at least two interesting applications, namely the implementation of variable-order de Bruijn graphs~\cite{BBGPS15} and streaming $k$-mer queries~\cite{ACDA2023}. We avoid further treatment of these applications here, and refer the reader to~\cite{BBGPS15,ACDA2023} for details. Our focus instead is on the efficient construction of the $\LCS$ array of a $k$-spectrum given its SBWT, which is also an interesting problem in its own right.

We are aware of little prior work on efficient $\LCS$ array construction for $k$-spectra. A naive approach is to expand the entire contents of the spectrum from the SBWT and scan it in colexicographical order. This requires $O(nk)$ time and $O(nk\log\sigma)$ bits of space, where $\sigma$ is the size of the alphabet of the $k$-mers in the spectrum. Bowe et al.~\cite{BBGPS15} make use of the $\LCS$ array for a $k$-spectra for variable order de Bruijn graphs, but do not address construction. Very recently, Conte et al.~\cite{conte2023computing} introduced LCS arrays for Wheeler graphs\footnote{Wheeler graphs are a class of graphs including de Bruijn graphs, that admit a generalization of the Burrows-Wheeler transform. The SBWT can be seen as a special case of the Wheeler graph indexing framework.}, but describe no construction algorithm, mentioning only in passing that a polynomial-time algorithm is possible. Prophyle, due to Salikov et al.~\cite{salikhov2017efficient} uses $k$-LCP information for sliding window queries on the BWT and describes an $O(nk)$ construction algorithm, where $n$ is the size of the full BWT, which can be orders of magnitude larger than the SBWT for repetitive datasets.

\paragraph{Contribution.} We describe three different algorithms for computing the $\LCS$ array of a $k$-spectrum from its SBWT. The first of these essentially decodes the $k$-mers of the $k$-spectrum in colexicographical order starting from their rightmost symbols in $k$ rounds, keeping track of when the suffixes become distinct using just $n(1 + \log\sigma)$ bits of side information (significantly less than the naive method mentioned above) and taking $O(nk)$ time overall. Our second approach is similar, but exploits the small DNA alphabet to decode multiple symbols per round with the effect of reducing computation and, importantly, CPU cache misses. Its running time is $O(cn + (k-c)n/c)$ overall with $O(nc\log\sigma + \sigma^c\log n)$ bits of extra space, where $c \le k$ is a parameter controlling a space-time tradeoff. Our final algorithm runs in time linear in $n$ --- independent of $k$ --- and while it is shaded by the second algorithm on smaller $k$, it becomes dominant as $k$ grows.

%...
%
%We can also dump the $k$-mers out of the matrix using $O(nk / \min(\log k, w / \sigma))$ time by using the super-alphabet technique, where $w$ is the width of the machine word.

\medskip

The remainder of this article is organized as follows. The next section sets notation and basic definitions. Sections~\ref{sec:basic}-\ref{sec:linear} then describe the three above-mentioned $\LCS$ array construction algorithms in turn. Section~\ref{sec:experiments} presents an experimental analysis of their performance in the context of a real pangenomic indexing task. Conclusions, reflections, and avenues for future work are then offered.

%%%%%%%%%%%%%%%%%%%%%%%%%%%%%%%%%%%%%%%%%%%%%%%%%%%%%%%%%%%%%%%%%%%%%%%%%%%%%%%%%
\section{Preliminaries}
\label{sec:prelims}

Throughout we consider a {\em string} $S = S[1..n] = S[1]S[2]\ldots S[n]$ on an integer alphabet $\Sigma$ of $\sigma$ symbols. In this article we are mostly interested in strings on the DNA alphabet, i.e. when $\Sigma = \{ A, C, C, T \}$. The \emph{colexicographic order} of two strings is the same as the lexicographic order of their reverse strings. The {\em substring} of $S$ that starts at position $i$ and ends at position $j$, $j \ge i$, denoted $S[i..j]$, is the string $S[i]S[i+1]\ldots S[j]$. If $i > j$, then $S[i..j]$ is the empty string $\varepsilon$. A suffix of $S$ is a substring with ending position $j = n$, and a prefix is a substring with starting position $i = 1$. We use the term $k$-mer to refer to a (sub)string of length $k$.

The following two basic definitions relate to $k$-spectra.

\begin{definition}
($k$-spectrum). The $k$-spectrum of a string $T$, denoted with $S_k(T)$, is the set of all distinct $k$-mers of the string $T$.
\end{definition}

\begin{definition} \label{def:prefix_set}
($k$-prefix set). The $k$-prefix set of a string $T$ is defined as the left-padded set of prefixes $P_k(T) = \{\$^{k-i}T[1..i] \; | \; i = 0, \ldots, k-1\}$, where $\$$ is a special character not found in the alphabet, that is smaller than all characters of the alphabet.
\end{definition}
The $k$-spectrum of a \emph{set} of strings $T_1, \ldots T_m$, denoted with $S_k(T_1, \ldots T_m)$, is defined as the union of the $k$-spectra of the individual strings. For example, consider the strings {\sf AGGTAAA} and {\sf ACAGGTAGGAAAGGAAAGT}. The 4-spectrum is the set \{{\sf GAAA}, {\sf TAAA}, {\sf GGAA}, {\sf GTAA}, {\sf AGGA}, {\sf GGTA}, {\sf AAAG}, {\sf ACAG}, {\sf GTAG}, {\sf AAGG}, {\sf CAGG}, {\sf TAGG}, {\sf AAGT}, {\sf AGGT}\}. Likewise, the $k$-prefix set $P_k(T_1, \ldots T_m)$ is the union of the $k$-prefix sets of the individual strings. In this case, the 4-prefix set is \{{\sf \$\$\$\$}, {\sf \$\$\$A}, {\sf \$\$AG} {\sf \$AGG}, {\sf \$\$AC}, {\sf \$ACA}\}.

\begin{definition}
\label{def:k_source_node}
($k$-source set). The $k$-source set $R_k(K)$ of a set of $k$-mers $K$ is the set $R_k(K) = \{x \in K \; | \; \not\exists y \in K \textrm{ such that } y[2..k] = x[1..k-1]\}$
\end{definition}
In our running example, the 4-source set of the 4-spectrum has just the 4-mer \{{\sf ACAG}\}. The \emph{extended $k$-spectrum} is the union of the spectrum and the $k$-prefix set of the $k$-source set, plus the $k$-mer $\$^k$ that is always added to avoid some corner cases.
\begin{definition}
\label{def:extended_k_spectrum}
(Extended $k$-spectrum). The extended $k$-spectrum $S'_k(T_1, \ldots, T_m)$ of a set of strings $T_1, \ldots, T_m$ is the set $S_k(T_1, \ldots, T_m) \cup P_k(R_k(T_1, \ldots, T_m)) \cup \{\$^k\}$
\end{definition}
We are now ready to define the SBWT. The definition below corresponds to the multi-SBWT definition of Alanko et al.~\cite{ACDA2023}.

\begin{definition} (Spectral Burrows-Wheeler transform, SBWT)
    Let $\{T_1, \ldots, T_m\}$ be a set of strings from an alphabet $\Sigma$. Let $x_i$ be the colexicographically $i$-th element of the extended $k$-spectrum $S_k'(T_1, \ldots, T_m)$ of size $n$. The SBWT of order $k$ is a sequence $X_1, X_2, \ldots X_n$ of subsets of $\Sigma$. The set $X_i$ is the empty set if $i > 1 \textrm{ and } x_{i-1}[2..k] = x_{i}[2..k]$, otherwise $X_i = \{c \in \Sigma \; | \; x_j[2..k]c \in S_k'(T_1, \ldots, T_m)\}$
\end{definition}

\begin{figure*}[t]
    \begin{minipage}{0.69\linewidth}
    \strut\vspace*{-\baselineskip}\newline
    \centering
    \includegraphics[width=1\textwidth]{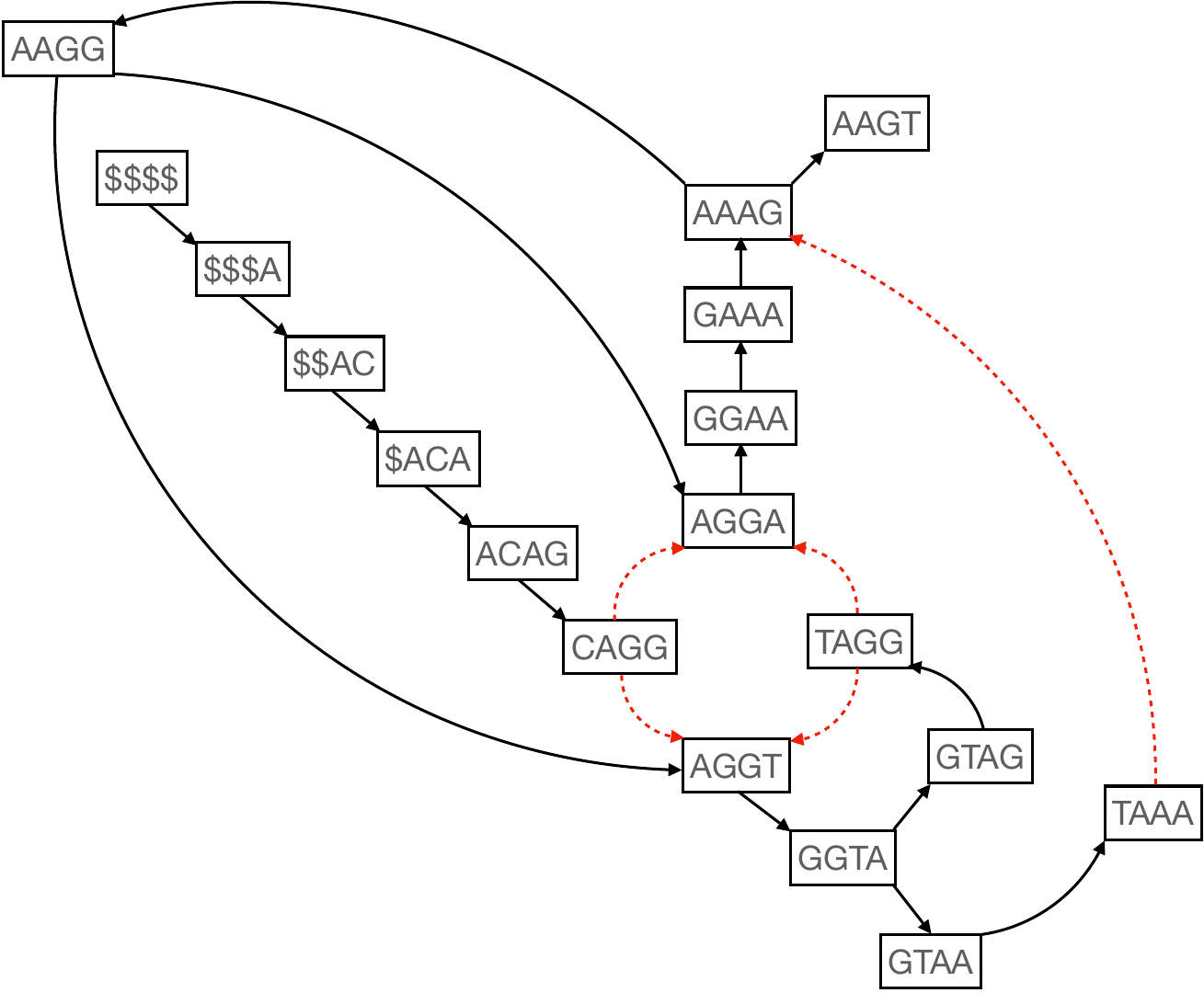}
    \end{minipage}
    %\hspace{0.1cm}
    \begin{minipage}{0.2\linewidth}
    \strut\vspace*{-\baselineskip}\newline
    \centering
        \resizebox{\textwidth}{!}{%

\begin{tabular}{lll}
\multicolumn{1}{c}{$k$-mers} & \multicolumn{1}{c}{{\em LCS}} & \multicolumn{1}{c}{{\em SBWT}} \\
\hline
{\sf \$\$\$\$} & -	& {\sf A}\\
{\sf \$\$\$A}	& 0	& {\sf C}\\
{\sf GAAA}	&1	& {\sf G}\\
{\sf TAAA}	&3	& {\sf -}\\
{\sf GGAA}	&2	& {\sf A}\\
{\sf GTAA}	&2	& {\sf A}\\
{\sf \$ACA}	&1	& {\sf G}\\
{\sf AGGA}	&1	& {\sf A}\\
{\sf GGTA}	&1	& {\sf A,G}\\
{\sf \$\$AC}	&0	& {\sf A}\\
{\sf AAAG}	&0	& {\sf G,T}\\
{\sf ACAG}	&2	& {\sf G}\\
{\sf GTAG}	&2	& {\sf G}\\
{\sf AAGG}	&1	& {\sf A,T}\\
{\sf CAGG}	&3	& {\sf -}\\
{\sf TAGG}	&3	& {\sf -}\\
{\sf AAGT}	&0	& {\sf -}\\
{\sf AGGT}	&2	& {\sf A}\\
\hline
\end{tabular}
}%resizebox
    \end{minipage}
    \caption{Left: The de Bruijn graph (with $k=4$) of the set of two strings $\lbrace${\sf AGGTAAA}, {\sf ACAGGTAGGAAAGGAAAGT}\}. Red dashed edges are pruned from the graph because the node they point to can be reached from another (black) edge. Right: The extended $k$-spectrum of the input strings in colexicographical order, together with the longest common suffix (LCS) array and the spectral Burrows-Wheeler transform (SBWT).} 
    \label{fig:lcs-sbwt-example}
\end{figure*}
\noindent The sets in the SBWT represent the labels of outgoing edges in the node-centric de Bruijn graph of the input strings, such that we only include outgoing edges from $k$-mers that have a different suffix of length $k-1$ than the preceding $k$-mer in the colexicographically sorted list. Figure~\ref{fig:lcs-sbwt-example} illustrates the SBWT and the associated de Bruijn graph. The addition of the $k$-prefix set of the $k$-source set is a technical detail necessary to make the transformation invertible and searchable.

There are many ways to represent the subset sequence of the SBWT \cite{ABPV23,ACDA2023}. In this paper, we focus on the \emph{matrix representation}. This representation is currently the most practical version known for small alphabets, and it is used e.g. in the $k$-mer pseudoalignment tool Themisto \cite{AVMP23}.

\begin{definition}\label{def:PlainMatrix}
(Plain Matrix SBWT)
The plain matrix representation of the SBWT sequence is a binary matrix $M$ with $\sigma$ rows and $n$ columns. The value of $M[i][j]$ is set to $1$ iff subset $X_j$ includes the $i^{th}$ character in the alphabet.  
\end{definition}
Figure~\ref{fig:matrix} illustrates the matrix SBWT of our running example. Lastly, we define the central object of interest in this paper: the LCS array of an SBWT:

\begin{definition}\label{def:LCS-array}
(Longest common suffix array, $\LCS$ array)
Let $\{T_1, \ldots, T_m\}$ be a set of strings and let $x_i$ denote the colexicographically $i$-th $k$-mer of $S_k'(T_1, \ldots, T_m)$. The LCS array is an array of length $|S_k'(T_1, \ldots, T_m)|$ such that $LCS[1] = 0$ and for $i > 1$, the value of $LCS[i]$ is the length of the longest common suffix of $k$-mers $x_i$ and $x_{i-1}$. 
\end{definition}
In the definition above, the empty string is considered a common suffix of any two $k$-mers, so the longest common suffix is well-defined for any pair of $k$-mers. Figure~\ref{fig:lcs-sbwt-example} illustrates the LCS array of our running example.

\begin{figure*}[t]
    \centering
    \resizebox{\textwidth}{!}{%
\begin{tabular}{lcccccccccccccccccc}
&
\multicolumn{1}{c}{{\sf \$\$\$\$}}&
\multicolumn{1}{c}{{\sf \$\$\$A}}&
\multicolumn{1}{c}{{\sf GAAA}}&
\multicolumn{1}{c}{{\sf TAAA}}&
\multicolumn{1}{c}{{\sf GGAA}}&
\multicolumn{1}{c}{{\sf GTAA}}&
\multicolumn{1}{c}{{\sf \$ACA}}&
\multicolumn{1}{c}{{\sf AGGA}}&
\multicolumn{1}{c}{{\sf GGTA}}&
\multicolumn{1}{c}{{\sf \$\$AC}}&
\multicolumn{1}{c}{{\sf AAAG}}&
\multicolumn{1}{c}{{\sf ACAG}}&
\multicolumn{1}{c}{{\sf GTAG}}&
\multicolumn{1}{c}{{\sf AAGG}}&
\multicolumn{1}{c}{{\sf CAGG}}&
\multicolumn{1}{c}{{\sf TAGG}}&
\multicolumn{1}{c}{{\sf AAGT}}&
\multicolumn{1}{c}{{\sf AGGT}}
\\
\hline
{\sf A} & 1 & 0 & 0 & 0 & 1 & 1 & 0 & 1 & 1 & 1 & 0 & 0 & 0 & 1 & 0 & 0 & 0 & 1 \\
{\sf C} & 0 & 1 & 0 & 0 & 0 & 0 & 0 & 0 & 0 & 0 & 0 & 0 & 0 & 0 & 0 & 0 & 0 & 0 \\
{\sf G} & 0 & 0 & 1 & 0 & 0 & 0 & 1 & 0 & 1 & 0 & 1 & 1 & 1 & 0 & 0 & 0 & 0 & 0 \\
{\sf T} & 0 & 0 & 0 & 0 & 0 & 0 & 0 & 0 & 0 & 0 & 1 & 0 & 0 & 1 & 0 & 0 & 0 & 0 \\
\hline
\end{tabular}}
    \caption{The binary matrix representation of the spectral Burrows-Wheeler transform (SBWT).} 
    \label{fig:matrix}
\end{figure*}

\section{Basic $O(nk)$-time $\LCS$ array construction}
\label{sec:basic}

Before describing how to compute the $\LCS$ array we are going to explain how the whole $k$-spectrum can be recovered from the SBWT. We can reconstruct the full $k$-spectrum from the binary matrix representation $M$ of the SBWT  with $\sigma$ rows and $n$ columns and the cumulative array $C$ included in the SBWT.
Since $k$-mers are colexicographically sorted, they are assembled back to front. First, the last character of each $k$-mer is retrieved based on the $C$ array. These last added characters will be accessed later and are then stored in an array $L$. In accordance with the LF mapping property, which holds also for the SBWT, the previous character of each $k$-mer is recursively retrieved until reaching length $k$ as follows: First, at each iteration, a copy $C'$ of the $C$ array is saved and the vector $P$ for storing the last propagated characters is initialised with a dollar symbol.
Then, each column $i$ of $M$ is scanned. If $M[c, i]$ is 1, the first free position of the $c$ block marked by the $C'$ array in $P$ is set to the character in $L[i]$. Since we are scanning every column in $M$, we do not need to issue rank queries, but it is instead sufficient to increase the counter $C'[c]$ by one. At the end of each iteration, the newly propagated characters are copied to $L$. Considering the de Brujin graph of the SBWT, with this procedure edge labels are propagated one step forward in the graph.

Calculating the $\LCS$ array from the SBWT is similar to the procedure described above. The $\LCS$ array is initialised as an array of zeros and it is updated at each round of $M$ scanning by checking the mismatches between two adjacent newly propagated characters. Once an entry of the $\LCS$ array is updated, it is never modified again.
Since for each character of the $k$-mers we need to traverse all columns of $M$ once, the whole $k$-spectra can be retrieved in $O(n \sigma k)$-time, where $n$ is the number of $k$-mers in the SBWT. %\todo{describe theoretical $O(n k \log \sigma)$ algorithm using the SSWT.}
Instead of scanning $M$ $k$ times, we could traverse the Subset Wavelet Tree of the string (see~\cite{ABPV23}) and issue a binary rank operation for every character in each subset. Repeating this for each $k$-mer character will result in the LCS construction in time $O(n k \log\sigma)$.  This reduces to $O(n k)$ assuming a constant $\sigma$. Computing the $\LCS$ array does not alter this time complexity.

\begin{algorithm}[!h]
\begin{algorithmic}
\small
\State $LCS \gets $ Array of length $n$ initialized to $0$
\State $mismatches \gets $ Array of length $n$ initialized to $0$ \Comment{positions set in $\LCS$}
\State $L \gets $ Array of length $n$, with $\sigma +1$ characters, initialized according to $C$
\For{$round =0 \ldots k-1$}
    
    \For{$ i = 1 \ldots n-1$} \Comment{LCS[1]=0 by definition}
        \If{$mismatches[i+1]=0$ \textbf{and} $L[i+1] \neq L[i]$} % L[i+1 has not been modified yet]
            \State $mismatches[i+1] \gets 1$
            \State $LCS[i+1] \gets round$ \Comment{store the longest match length}
        \EndIf
    \EndFor
    
    \State $P \gets $ Array of length $n$ initialized to $\$$
    \State $C' \gets$ copy of the $C$ array
    \For{$i =1 \ldots n$}
        \For{$c \in \Sigma$}
        \If{$M[c,i] = 1$}
        \State $C'[c]\gets C'[c] + 1$ 
        \State $P[C'[c]] \gets L[i]$
        \EndIf
        \EndFor
    \EndFor

    \State $L \gets P$
\EndFor
\State \Return $\LCS$
\end{algorithmic}
\caption{Basic $\LCS$ array construction in $O(nk)$ time. \\ 
\textbf{Input}: SBWT matrix $M$ with $n$ columns and $\sigma$ rows, $\Sigma = \{1,\ldots,\sigma\}$ and $C$ array. \\ %
% \textbf{Input}: Character $c$ from an alphabet $\Sigma = \{1,\ldots,\sigma\}$ and an index $i$. \\ %
\textbf{Output}: $k$-bounded $\LCS$ array.}
\label{alg:subsetWT_rank_query}
\end{algorithm}

\section{Faster construction via super-alphabet techniques}
\label{sec:superalpha-lol}

\begin{figure*}[t]
    \centering
    \resizebox{\textwidth}{!}{%
\begin{tabular}{llccccccccccccccccccccccccccccccccccc}
\hline
 & & {\scriptsize 0} & {\scriptsize 1} & {\scriptsize 2} & {\scriptsize 3} & {\scriptsize 4} & {\scriptsize 5} & {\scriptsize 6} & {\scriptsize 7} & {\scriptsize 8} & {\scriptsize 9} & {\scriptsize 10} & {\scriptsize 11} & {\scriptsize 12} & {\scriptsize 13} & {\scriptsize 14} & {\scriptsize 15} & {\scriptsize 16} & {\scriptsize 17} & {\scriptsize 18} & {\scriptsize 19} & {\scriptsize 20} & {\scriptsize 21} & {\scriptsize 22} & {\scriptsize 23} & {\scriptsize 24} & {\scriptsize 25} & {\scriptsize 26} & {\scriptsize 27} & {\scriptsize 28} & {\scriptsize 29} & {\scriptsize 30} & {\scriptsize 31} & {\scriptsize 32} & {\scriptsize 33} & {\scriptsize 34} \\
{\em V} & & {\sf A} & {\sf C} & {\sf G} & {\sf A} & {\sf A} & {\sf G} & {\sf A} & {\sf A} & {\sf G} & {\sf A} & {\sf A} & {\sf G} & {\sf T} & {\sf G} & {\sf G} & {\sf A} & {\sf T} & {\sf A} \\
{\em B} & & 1 & 0 & 1 & 0 & 1 & 0 & 1 & 1 & 0 & 1 & 0 & 1 & 0 & 1 & 0 & 1 & 0 & 0 & 1 & 0 & 1 & 0 & 0 & 1 & 0 & 1 & 0 & 1 & 0 & 0 & 1 & 1 & 1 & 1 & 0 \\
\hline
\end{tabular}
} %{\sf \$}
    \caption{The concatenated representation of the spectral Burrows-Wheeler transform (SBWT) used by the super-alphabet-based LCS construction algorithm.} 
    \label{fig:matrix}
\end{figure*}

The super-alphabet techniques described here are based on first decoding a $c$-symbol suffix of each $k$-mer using the previous algorithm in $O(cn)$ time and subsequently computing the remaining information in $O(1 + (k-c)/c)$ rounds and $O(cn + (k-c)n/c)$ time overall with $O(n)$ extra space.
%\todo{Check $(k-c)$ is correct. it is correct if $k$ is divisible by $c$, if not we need to go from $c$ up to $k+c-1$.}
%\todo{Should we say $\min(\sigma^c, n)$ instead of $\sigma^c$ since there can't be more than $n$ super-characters.}
Given $c=2$, the algorithm first replicates the basic one up to the computation of the last 2 characters of each $k$-mer as well as their $\LCS$ values. At this point, the $2$ last symbols of the $i^{th}$ $k$-mer, $P[i]$ and $L[i]$, are combined to create a super-character (or meta-character) $P[i]\cdot L[i]$ which is stored in $L[i]$. 
%\todo{ Is this obvious? A super-character is a character belonging to a super-alphabet created by pairing characters from the previous alphabet.} 
A new C array is then generated from the alphabet of super-characters. The following super-characters for each $k$-mer are then retrieved as in the basic algorithm. The only difference is that in the present case, the algorithm uses the concatenated representation of the SBWT of super-characters instead of the plain matrix representation.   
The concatenated representation of the SBWT sequence\footnote{A similar but different structure is described in~\cite{ACDA2023}.} consists of a concatenation of the subsets characters, stored in a vector $V$, and an encoding of the subsets sizes stored in a bitvector $B$. In further detail, let $S(X_i)$ be the concatenation of characters in the subset $X_i$, then $V=S(X_i)\cdot S(X_2)\cdot S(X_n)$. No symbol will be stored in $V$ if $X_i$ is the empty set. The empty sets are represented in $B= 1\cdot 0^{|S(X_1)|}\cdot 1\cdot 0^{|S(X_2)|} \cdots 1\cdot 0^{|S(X_n)|}$. % unary? 
The concatenated representation of a $c$-super-alphabet, $V'$ and $B'$, can be obtained from $V$ and $B$, the concatenated representation of the $c/2$-(super-)alphabet. % would it make sense to say 2c and c?
$V'$ is filled in, scanning $V$, with $V[j]$ where $0\geq j \leq |V|$ concatenated with the characters in the subset $X$ marked by the $C$ array entry of $V[j]$ in $V$. For each character in $V$, $1\cdot0^{|X|}$ is appended to $B'$. No rank nor select queries are necessary as it is sufficient to update a copy of the $C$ array.
Considering the de Brujin graph of the SBWT, to create a super-concatenated representation edge labels are propagated one step backward in the graph.

Similarly to the basic algorithm% and always in compliance with the LF mapping property
, the preceding super-character of each $k$-mer is recursively retrieved until reaching length $k$ as follows: First, at each iteration, a copy of the super $C$ array is stored and $P$ is initialised with the smallest super-character $\$^{c}$. Then $V'$ is scanned keeping track of the number of subsets encountered with a counter $v$ which is increased by $1$ if $B[i+v]=1$. %If $B[i+v]$ is $0$, the first position of the $V'[i]$ super-character block marked by the $C$ array in $P$ is set to $L[i]$.
If $B[i+v] = 0$, $L[i]$ is assigned to $P$ at the index corresponding to the position of the $V'[i]$ super-character block marked by the $C$ array.
As for the basic alphabet, since every subset is inspected in order, there is no need to issue rank queries, but it is instead sufficient to increase the copied $C'$ counter for $V'[i]$ by one. At the end of each iteration, the newly propagated super-characters are stored in $L$. Since we are skipping nodes in the graph, the iteration number $r$ goes from $c$ to at most $k+c-1$ with steps of size $c$.

The $\LCS$ array using super-characters is computed by checking first the presence of mismatches in the rightmost single characters with an appropriate mask and only if no mismatch is found, subsequent characters are checked. The $\LCS$ is then updated accordingly. Given a super-character with $c=2$ at index $i$ as $c2\cdot c1$, the algorithm compares first $c1[i]$ and $c1[i-1]$. In the presence of a mismatch $\LCS$ would be updated to the iteration number $r$. If $c2[i]\neq c2[i-1]$, $\LCS[i]= r + 1$ since 1 is, in this case, the number of matches found in the characters of the super-character. If on the contrary, $c2[i]=c2[i-1]$, the $\LCS$ could not be updated yet. The algorithm never checks more characters than necessary as it stops at the first encountered mismatch.

\section{Construction in linear time}
\label{sec:linear}
Our linear-time algorithm can be seen as a generalization of the linear-time LCP algorithm of Beller et al. \cite{beller2013computing} from the regular BWT to the SBWT. When the input is the spectrum of a single string and $k$ approaches $n$, the SBWT coincides with the BWT of the reverse of the input\footnote{Assuming the input to the BWT is terminated with a \$-symbol, and there is an added \$-edge from the last $k$-mer of the input to the root of the SBWT graph.}, and both algorithms perform the same iteration steps.

The algorithm fills in the LCS in increasing order of the values. The main loop has $k$ iterations, such that iteration $i$ fills in LCS values that are equal to $i-1$. Values that are not yet computed are denoted with $\bot$. 

We denote the \emph{colexicographic interval} of string $\alpha$ with $[\ell, r]_\alpha$, where $\ell$ and $r$ respectively are the colexicographic ranks of the smallest and largest $k$-mer in the SBWT that have $\alpha$ as a suffix. The right extensions of interval $[\ell, r]_\alpha$, denoted with EnumerateRight($\ell, r$), are those characters $c$ such that $\alpha c$ is a suffix of at least one $k$-mer in the SBWT. The interval of right extension $c$ from $[\ell, r]_\alpha$, denoted with ExtendRight($\ell, r, c$), can be computed using the formula $[2 + C[c] + rank_c(\ell-1), 1 + C[c] + rank_c(r)]_{\alpha c}$, where the rank is over the subset sequence of the SBWT \cite{ACDA2023}, and $C[c]$ is the number of characters in the SBWT that are smaller than $c$.

The input to iteration $i$ is a list of colexicographic intervals of substrings of length $i-1$. For each interval $[\ell, r]_{\alpha}$ in the list, the algorithm computes all right-extensions $[\ell', r']_{\alpha c}$. If $LCS[r'+1]$ is not yet filled yet, the algorithm sets $LCS[r'+1] = i-1$ and adds $[\ell', r']_{\alpha c}$ to the list of intervals for the next round. Otherwise, $LCS[r'+1]$ is not modified and interval $[\ell', r']_{\alpha c}$ is not added to the next round. Algorithm \ref{alg:linear} lists the pseudocode. The algorithm is designed so that at the end, every value of the LCS array has been computed.

\begin{algorithm}[th]
\begin{algorithmic}[1]
\small
\State $LCS \gets $ Array of length $n$ initialized to $\bot$
\State $LCS[1] \gets 0$ \Comment{By definition.}
\State $I \gets $ $([1, n])$ \Comment{List of intervals for current round.}
\State $I' \gets $ $([1, 1])$\Comment{List of intervals for the next round. Here interval of \$}
\For{$i = 1..k$} 
\While{$|I| > 0$}
\State $[\ell, r] \gets $ Pop $I$
\For{$c \in$ EnumerateRight($\ell, r$)}
    \State $[\ell', r'] \gets $ ExtendRight($\ell, r, c$)
    \If{$r' < n $ and $ LCS[r'+1] = \bot$}
        \State $LCS[r'+1] \gets i-1$ \label{algline:new_LCS_value}
        \State Push $[\ell', r']$ to $I'$
    \EndIf
\EndFor
\EndWhile
\State $I \gets I'$
\State $I' \gets$ Empty list
\EndFor
\State \Return LCS
\end{algorithmic}
\caption{Construction in $O(n\log \sigma)$ time. \\ 
\textbf{Input}: SBWT with support for EnumerateRight and ExtendRight. \\
\textbf{Output}: $k$-bounded $\LCS$ array.}
\label{alg:linear}
\end{algorithm}

\subsection{Correctness}\label{sec:linear_correctness}

To prove the correctness of the algorithm, we introduce the concept of an \emph{L-interval}. A colexicographic interval $[\ell, r]_\alpha$ is called an L-interval iff it is the longest colexicographic interval of a string with interval endpoint $r$. In case there are multiple strings with the same interval $[\ell, r]$, then the $\alpha$ in the subscript of the notation is the shortest string with this interval. The number of L-intervals is clearly $O(n)$ because each L-interval has a distinct endpoint. LCS array can be derived from the L-intervals as follows:

\begin{lemma}\label{lemma:L_interval_to_lcs}
    If $[\ell, r]_{c \alpha}$ is an L-interval, with $\alpha \in \Sigma^*$ and $c \in \Sigma$, then $LCS[r+1] = |\alpha|$
\end{lemma}
\begin{proof}
    It must be that $LCS[r+1] < |c \alpha|$ because otherwise the $k$-mer with colexicographic rank $r+1$ should have been included in the interval $[\ell, r]_{c \alpha}$. It must be that $LCS[r+1] \geq |\alpha|$ because otherwise the interval of $\alpha$ also has endpoint $r$, which means that $c \alpha$ is not the shortest string with interval ending at $r$, contradicting the initial assumption.
\end{proof}
The L-intervals form a tree, where the children of $[\ell, r]_\alpha$ are the single-character right-extensions $[\ell', r']_{\alpha c}$ that are L-intervals. The Lemma below implies that every L-interval is reachable by right extensions by traversing only L-intervals from the interval of the empty string:

\begin{lemma}\label{lemma:L_interval_tree}
Let $\alpha c$ be a substring of the input such that $\alpha \in \Sigma^*$ and $c \in \Sigma$. If $[\ell, r]_{\alpha c}$ denotes an L-interval, then $[\ell', r']_{\alpha}$ is an L-interval.
\end{lemma}

\begin{proof}
Suppose for a contradiction that the Lemma does not hold. Then there exists an L-interval interval $[x, r']_{\beta}$ with $x \leq \ell'$ such that $\beta$ is a proper suffix of $\alpha$. Then by the SBWT right extension formula, the interval $[\ell'', r'']_{\beta c}$ is such that $r'' = r$ and $\ell'' \leq \ell$. It can't be that $\ell'' = \ell$, or otherwise $\alpha c$ was not the shortest string with interval $[\ell,r]$, and it can't be that $\ell'' < \ell$ because then the starting point $\ell$ was not minimal for end point $r$. In both cases we have a contradiction, which proves the claim.
\end{proof}
We can now prove the correctness and the time complexity of the algorithm:

\begin{theorem}
    Given an SBWT having $n$ subsets of alphabet $\Sigma$ with $|\Sigma| = O(1)$, Algorithm \ref{alg:linear} correctly computes every value of the LCS array in time $O(n)$.
\end{theorem}

\begin{proof}
The algorithm traverses the L-interval tree in breadth-first order by right-extending from the empty string and visiting the shortest string representing each L-interval. Whenever the algorithm comes across an interval $[\ell', r']$ such that $LCS[r'+1]$ is already set, we know that endpoint $r'$ has already been visited before with a string shorter than the current string, so either $[\ell', r']$ is not an L-interval or the current string is not the shortest representative of it, so we can ignore it. By Lemma~\ref{lemma:L_interval_tree}, the shortest representative string of every L-interval is reachable this way. There is guaranteed to be an L-interval for every endpoint $r$ because there is at least a singleton colexicographic interval to every endpoint. Therefore, every value of the LCS array is eventually computed, and by Lemma~\ref{lemma:L_interval_to_lcs}, every computed value is correct. Since the number of L-intervals is $O(n)$, and EnumerateRight and ExtendRight can be implemented in constant time for a constant-sized alphabet, the total time is $O(n)$.
\end{proof}
For small alphabets, the call to EnumerateRight can be replaced by a process that tries all $\sigma$ possible right extensions. In this case, it is enough to track only interval endpoints, halving the space and number of rank queries required.

\section{Experimental Evaluation}
\label{sec:experiments}

\paragraph{Experimental Setup.} All our experiments were conducted on a machine with four 2.10\,GHz Intel Xeon E7-4830 v3 CPUs with 12 cores each for a total of 48 cores, 30\,MiB L3 cache, 1.5\,TiB of main memory, and a 12\,TiB serial ATA hard disk.
The OS was Linux (Ubuntu 18.04.5 LTS) running kernel 5.4.0-58-generic. The compiler was \texttt{g++} version 10.3.0 and the relevant compiler flags were \texttt{-O3} and \texttt{-DNDEBUG} (\texttt{-march=native} was not used). All runtimes were recorded by instrumenting the code with calls to \texttt{ std::chrono}. The peak memory (RSS) was measured using the \texttt{getrusage} Linux system call. C++ source code of the implementations tested is available upon request from the authors.

\paragraph{Datasets.} We experiment on three data sets representing different types of sequencing data found in genomics applications:
\begin{enumerate} 
\item A pangenome of 3682 E. coli genomes. The data was downloaded during the year 2020 by selecting a subset of 3682 assemblies listed in \url{ftp://ftp.ncbi.nlm.nih.gov/genomes/genbank/bacteria/assembly_summary.txt} with the organism name ``Escherichia coli'' with date before March 22, 2016. The resulting collection is available at \url{zenodo.org/record/6577997}. It contains 745,409 sequences of a total length 18,957,578,183.
\item The human reference genome version GRCh38.p14, available at \url{https://www.ncbi.nlm.nih.gov/assembly/GCF_000001405.40}. It contains 705 sequences of total length 3,298,430,636.
\item A set of 34,673,774 paired-end Illumina HiSeq 2500 reads each of length 251 sampled from the human gut (SRA identifier ERR5035349) in a study on irritable bowel syndrome and bile acid malabsorption~\cite{jeffery2020differences}. The total length of this data set is 8,703,117,274 bases.
\end{enumerate}
We focus solely on genomic data as that is currently the main application of the SBWT. 
%Table~\ref{table:datasets} lists some key statistics on the datasets. 
The constructed index structures include both forward and reverse DNA strands. We experiment with values $k = 16,32,48,64,80,96,112,128$ and $255$. For the metagenomic reads, the maximum value used was 251 since this is the length of the reads.
Fig.~\ref{fig:kmers} shows a plot of the number of distinct $k$-mers for varying $k$.

\paragraph{Algorithms.} The basic and linear algorithms are implemented on top of the matrix representation of the SBWT. In the linear algorithm, we apply the observation mentioned at the end of Section~\ref{sec:linear_correctness} and only track interval end points.

The super-alphabet algorithm (labelled {\sf SA-2} in the plots) first constructs the concatenated representation from the matrix representation and operates on it alone after the initial round of alphabet expansion. We experimented only with a super-alphabet of size 2, and leave a more detailed exploration, including larger super-alphabets, for future work.

\paragraph{Results.} Fig.~\ref{fig:runtime} shows on the top the runtime of each algorithm as a function of the $k$-mer size for each of the three data sets. We observe that the super-alphabet algorithm is consistently faster than the basic and linear algorithms until $k$ reaches 128, after which the linear algorithm is clearly fastest --- roughly three times faster than the basic algorithm on the E.coli dataset.

Memory usage for the algorithms is displayed at the bottom of Fig.~\ref{fig:runtime}. The super-alphabet algorithm uses significantly more memory than the other two, which is partly attributable to its use of the concatenated representation of the SBWT, which it must first build from the matrix representation, increasing peak memory. Moreoever, it uses a larger data type to hold the current column of the SBWT matrix (a 16-bit word per element instead of an 8-bit one used in the basic algorithm). 
%\todo{char, 16 bits, 64 bits?} 
In comparison, the basic and linear implementations use startingly little memory, which may make them preferable on systems where memory is scarce.

%\begin{table*}[t!]
%\centering
%\small
%\begin{tabular}{lrrr}
%\hline
%                    & \textbf{Sequences} & \textbf{Total length} & \textbf{Unique 31-mers} \\ \hline
%\textbf{E. coli} & 745,409 & 18,957,578,183 &  170,648,610\\ 
%\textbf{Human} & 705 & 3,298,430,636  & ? \\ 
%\textbf{Metagenome} & 34,673,774 & 8,703,117,274 & 2,761,523,935\\ \hline
%\end{tabular}
%\caption{\small Statistics on the raw genomic datasets used in our experiments. A $k$-mer is considered equal to its reverse complement in the $k$-mer counts.} \label{table:datasets}
%\end{table*}

\begin{figure*}[th]
    \centering
    \includegraphics[width=0.450\textwidth]{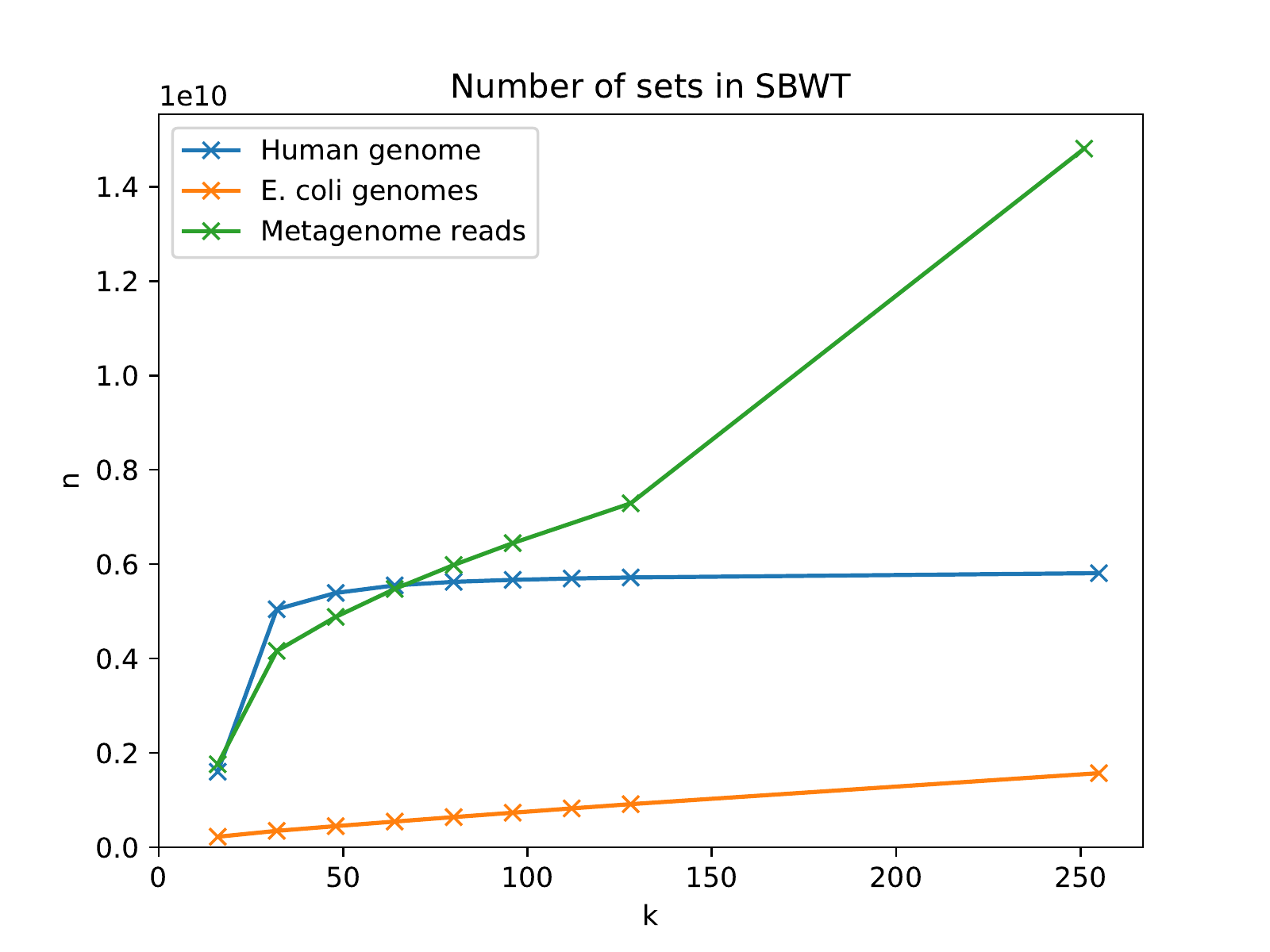}
    \caption{The number of sets in the SBWT (approximately equal to the number of $k$-mers) in each dataset for various $k$ used in our experiments.} 
    \label{fig:kmers}
    %\todo{Fix axes of this plot}
\end{figure*}

% This does not fit
%\begin{figure}[thb]
%  \centering
%  \includegraphics[width=.34\textwidth]{plots/coli_time_by_k.pdf}
%  \includegraphics[width=.34\textwidth]{plots/coli_mem_by_k.pdf}
%  \includegraphics[width=.34\textwidth]{plots/human_time_by_k.pdf}
%  \includegraphics[width=.34\textwidth]{plots/human_mem_by_k.pdf}
%  \includegraphics[width=.34\textwidth]{plots/metagenome_time_by_k.pdf}
%  \includegraphics[width=.34\textwidth]{plots/metagenome_mem_by_k.pdf}
%  \caption{Left column: Runtime as a function of $k$. Right column: Peak memory usage.}
%\end{figure}

% This does not fit
%\begin{figure}
%    \centering
%    \subfigure{
%        \includegraphics[width=0.35\textwidth]{plots/coli_time_by_k.pdf}
%        \includegraphics[width=0.35\textwidth]{plots/coli_mem_by_k.pdf}
%    }
%    \subfigure{
%        \includegraphics[width=0.35\textwidth]{plots/human_time_by_k.pdf}
%        \includegraphics[width=0.35\textwidth]{plots/human_mem_by_k.pdf}
%    }
%    \subfigure{
%        \includegraphics[width=0.35\textwidth]{plots/metagenome_time_by_k.pdf}
%        \includegraphics[width=0.35\textwidth]{plots/metagenome_mem_by_k.pdf}
%    }
%        \caption{Runtime and memory usage of LCS array construction algorithms versus $k$.}
%\end{figure}

\begin{figure*}[th]
\begin{minipage}[b]{0.32\linewidth}
\centering
\includegraphics[width=\textwidth]{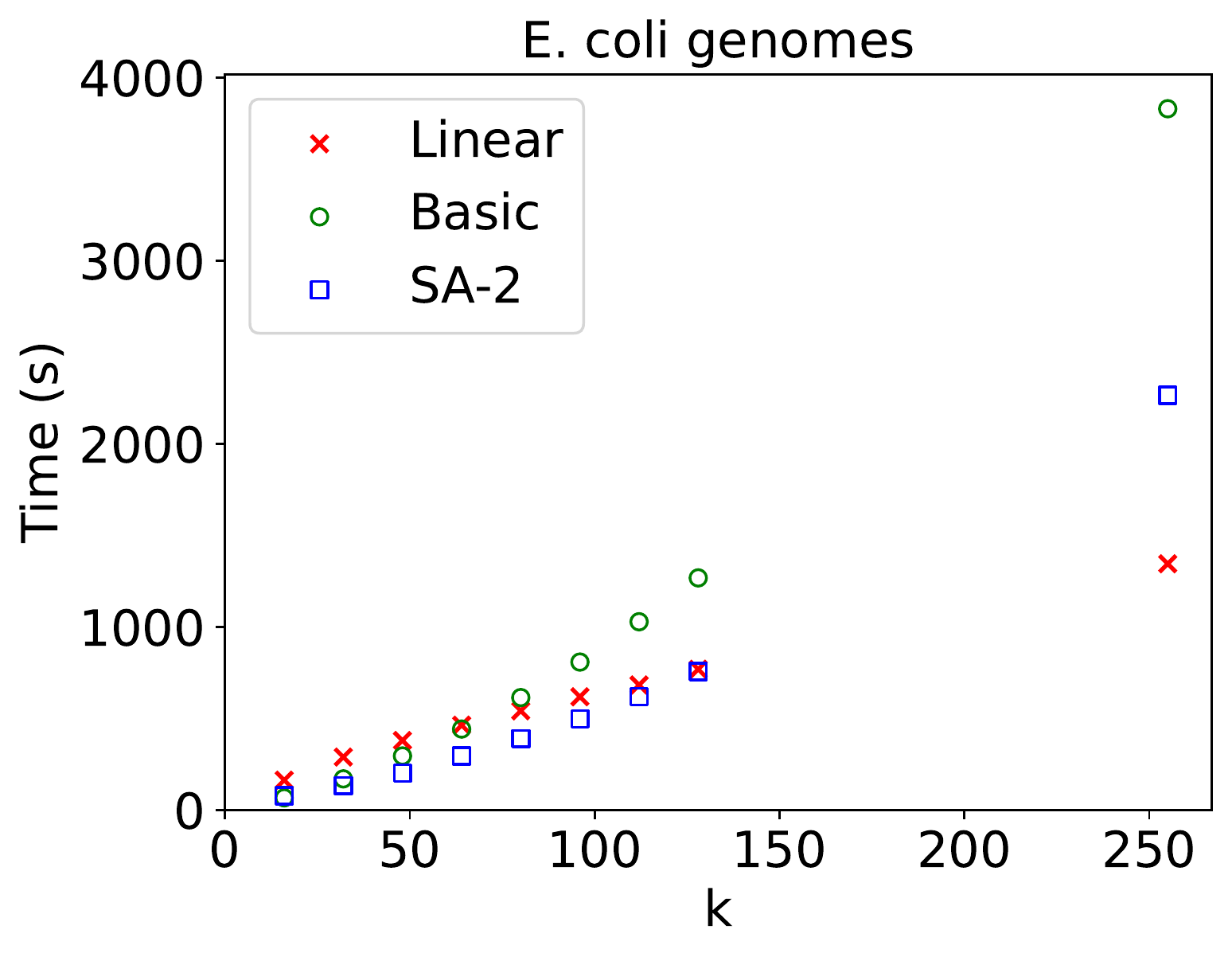}
\end{minipage}
\begin{minipage}[b]{0.32\linewidth}
\centering
\includegraphics[width=\textwidth]{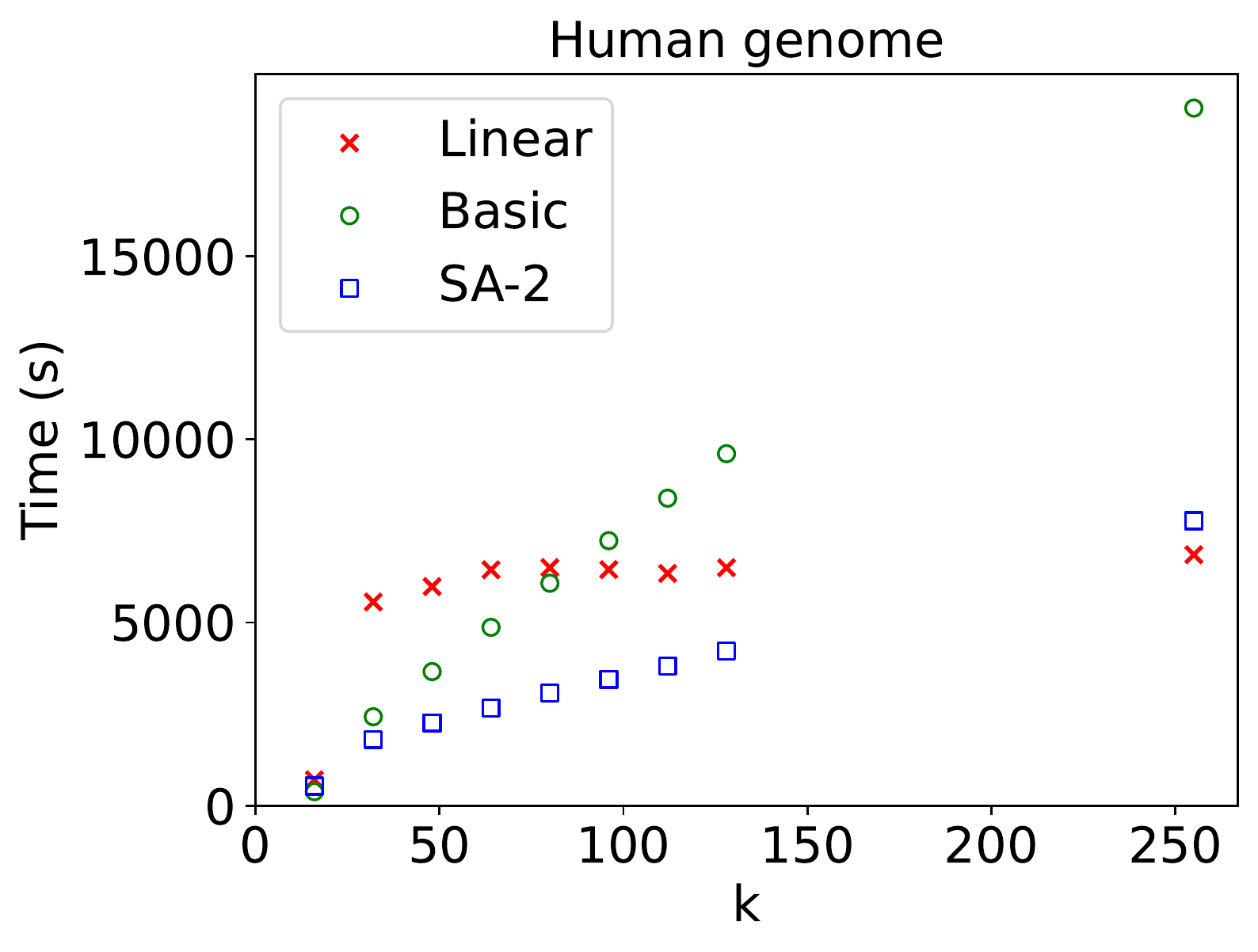}
\end{minipage}
\begin{minipage}[b]{0.32\linewidth}
\centering
\includegraphics[width=\textwidth]{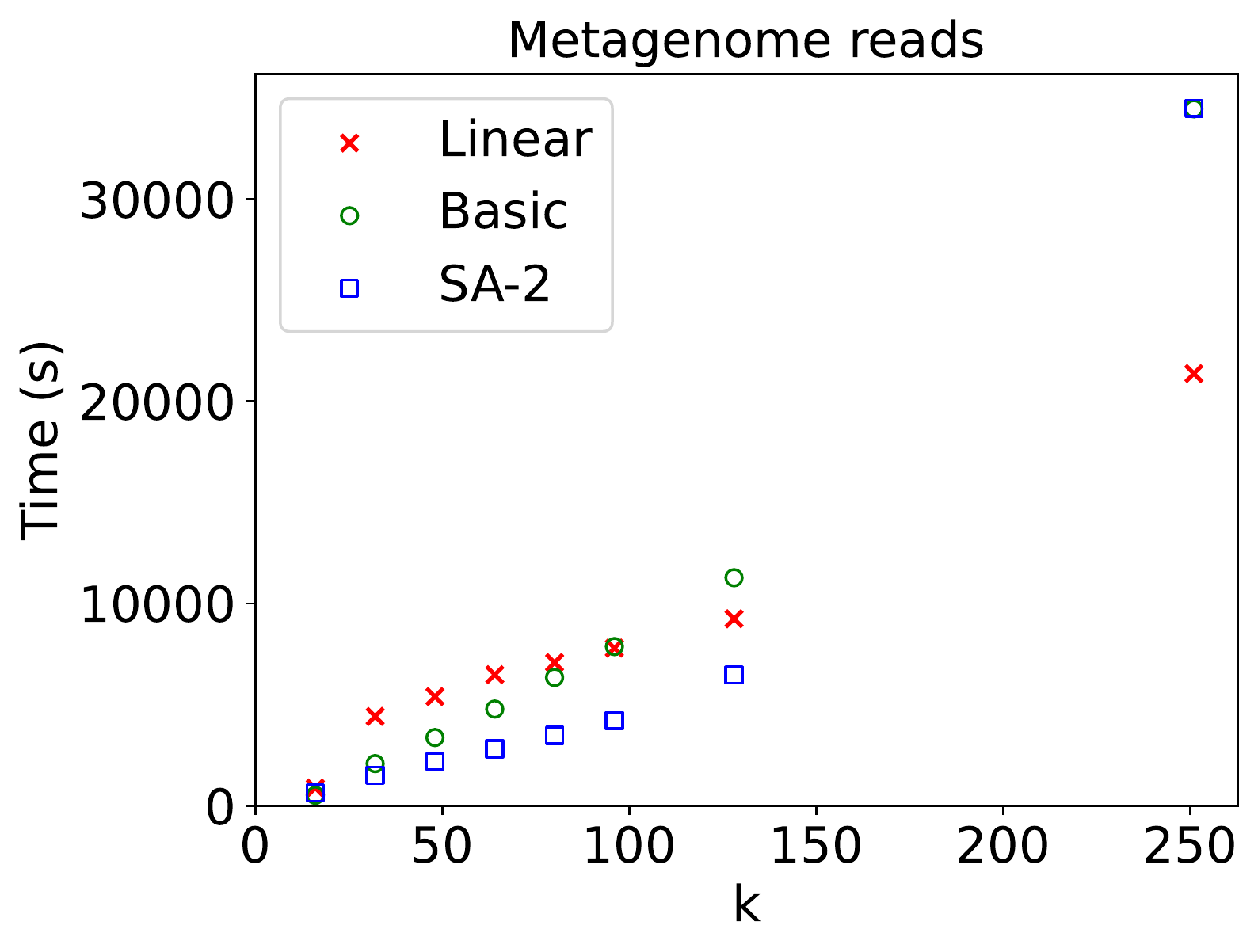}
\end{minipage}

\begin{minipage}[b]{0.32\linewidth}
\centering
\includegraphics[width=\textwidth]{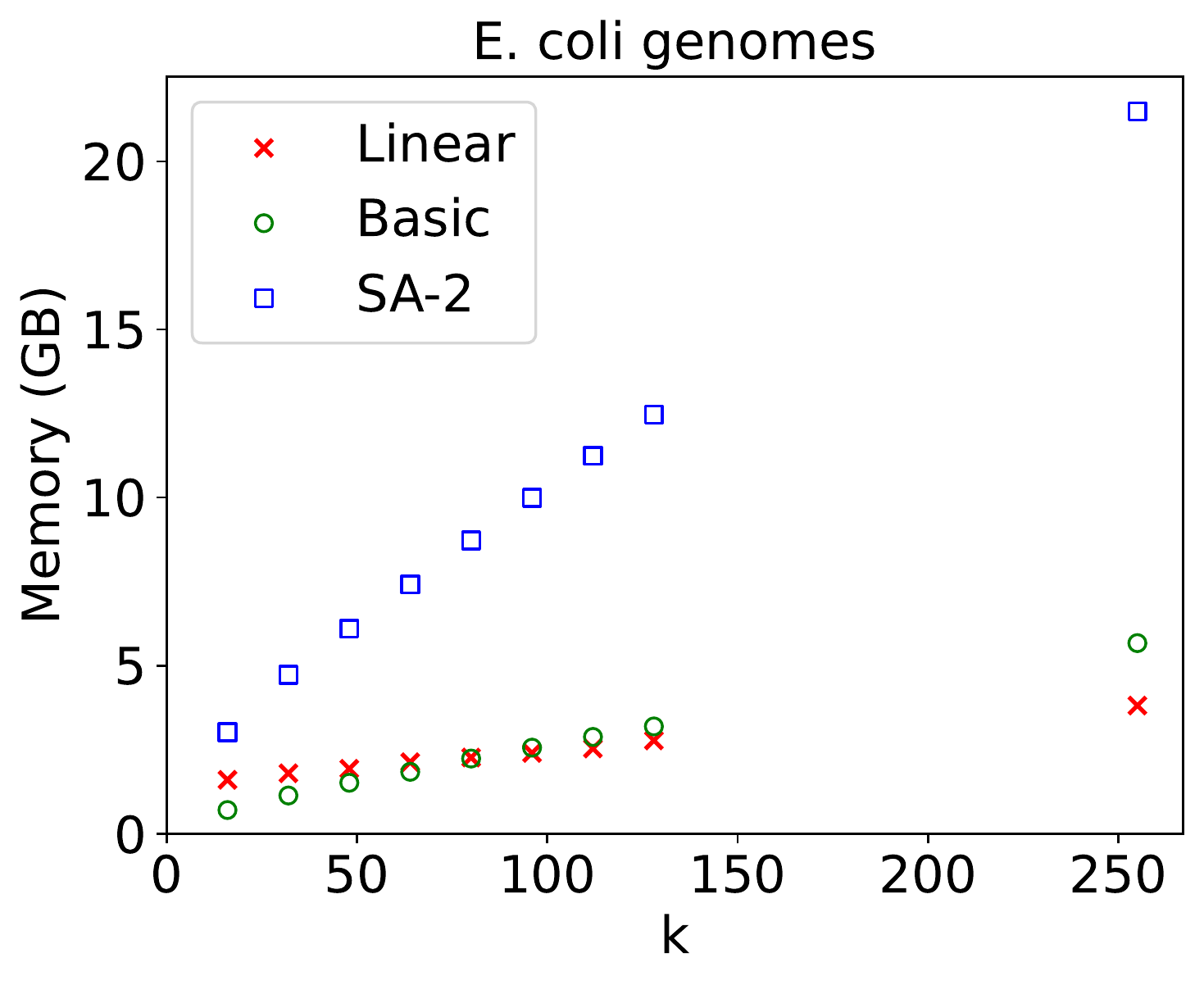}
\end{minipage}
\begin{minipage}[b]{0.32\linewidth}
\centering
\includegraphics[width=\textwidth]{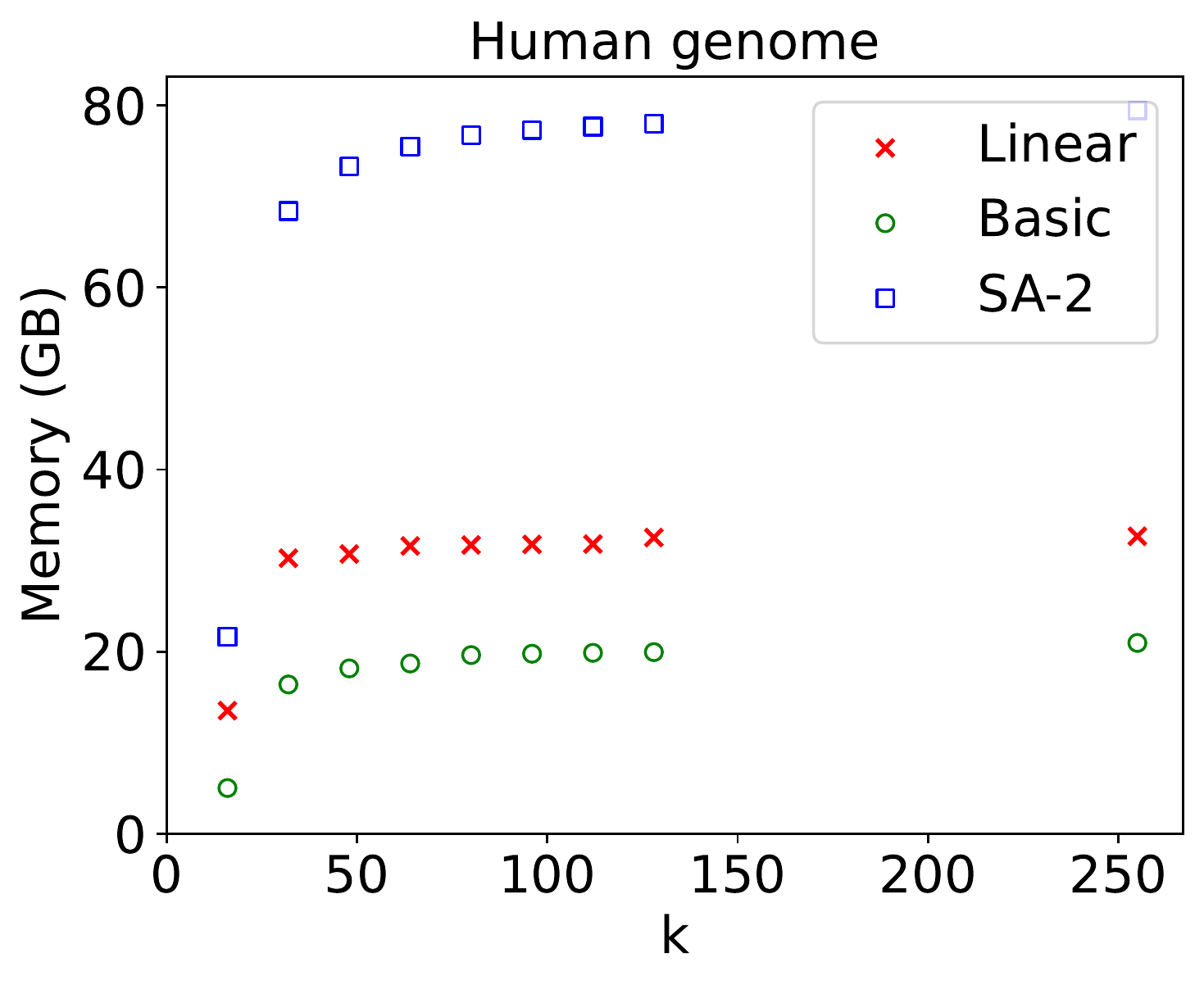}
\end{minipage}
\begin{minipage}[b]{0.32\linewidth}
\centering
\includegraphics[width=\textwidth]{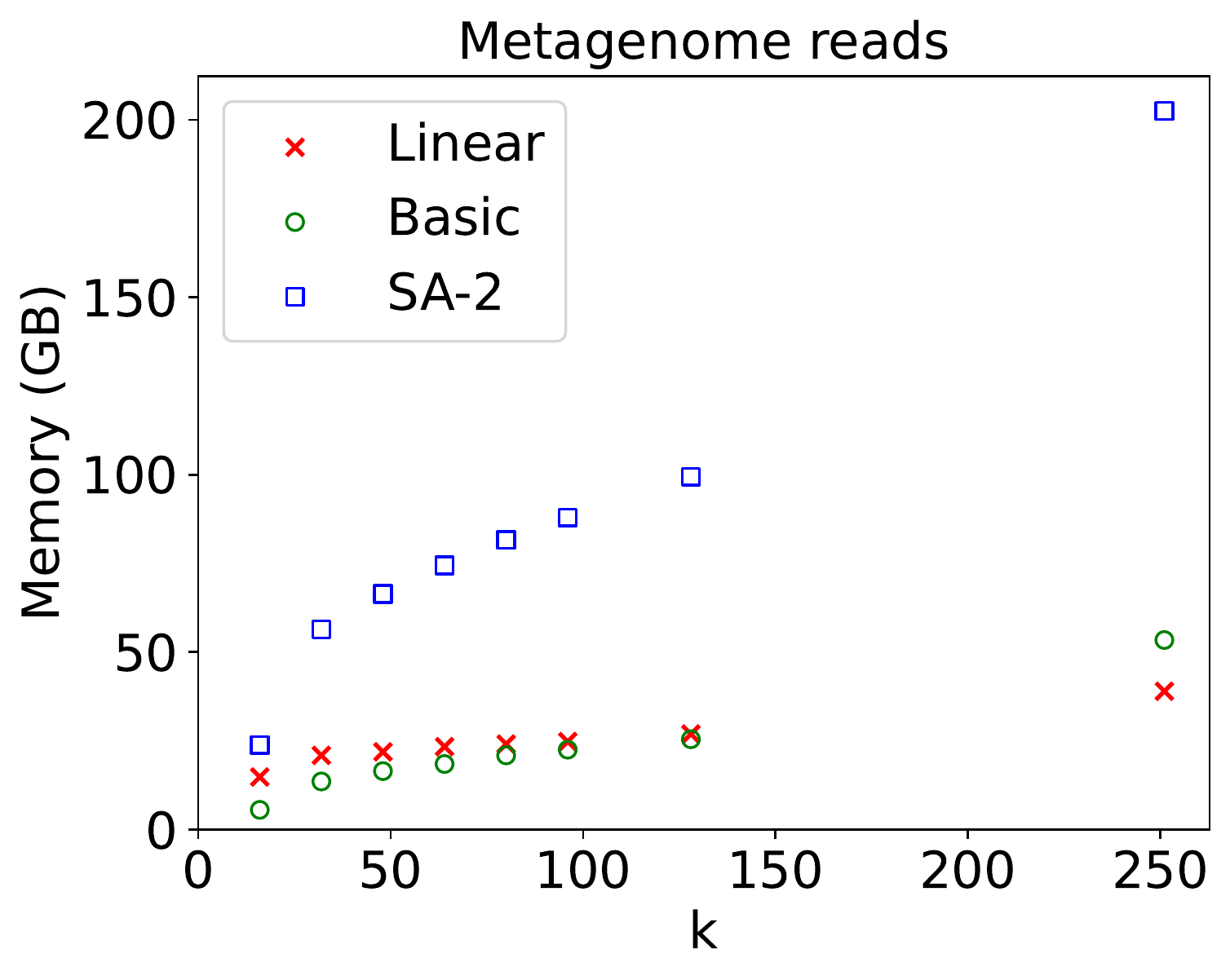}
\end{minipage}

    \caption{Runtime and memory usage of LCS array construction algorithms versus $k$.} 
    \label{fig:runtime}
\end{figure*}

%\begin{figure*}[hpbt]
%\begin{minipage}[b]{0.32\linewidth}
%\centering
%\includegraphics[width=\textwidth]{plots/coli_mem_by_k.pdf}
%\end{minipage}
%\begin{minipage}[b]{0.32\linewidth}
%\centering
%\includegraphics[width=\textwidth]{plots/human_mem_by_k.pdf}
%\end{minipage}
%\begin{minipage}[b]{0.32\linewidth}
%\centering
%\includegraphics[width=\textwidth]{plots/metagenome_mem_by_k.pdf}
%\end{minipage}
%    \caption{Peak memory (ordinate) versus $k$ (abscissa) for LCS array construction algorithms.} 
%    \label{fig:memory}
%\end{figure*}

\section{Concluding Remarks}

We have explored the design space of longest common suffix array construction algorithms for $k$-spectra. In particular, we have described two algorithms that, on real genomic datasets, significantly outperform our baseline $O(nk)$-time, $O(n)$ space approach. The first exploits the smaller nucleotide alphabet to form metacharacters and reduce the number of rounds needed by the basic algorithm. The second takes linear time (assuming a constant-size alphabet) by computing the LCS values in a special order and also performs well in practice, especially when $k$ is large.

All our algorithms have some dependency on $\sigma$ and we leave removing this as an open problem. From a practical point of view, it would be interesting to develop parallel algorithms that may further accelerate LCS array construction on large data sets.
\newpage
\bibliography{biblio}

\end{document}